\documentclass[12pt]{amsart}

\usepackage[english]{babel}
\usepackage{times,bm,amsfonts,amsmath,amssymb,dsfont}
\usepackage{pdfsync}
\usepackage{graphicx}
\usepackage[babel=true]{csquotes}
\usepackage{mathabx}
\usepackage{pdfsync}
\usepackage{enumerate}

\usepackage{color}

\makeatletter
\def\section{\@startsection{section}{1}\z@{.9\linespacing\@plus\linespacing}%
  {.7\linespacing} {\fontsize{13}{15}\selectfont\scshape\centering}}
\def\paragraph{\@startsection{paragraph}{4}%
  \z@{0.3em}{-.5em}%
  {$\bullet$ \ \normalfont\itshape}}
\makeatother

\definecolor{gr}{rgb}   {0.,   0.69,   0.23 }
\definecolor{bl}{rgb}   {0.,   0.5,   1. }
\definecolor{mg}{rgb}   {0.85,  0.,    0.85}
\definecolor{or}{rgb}   {0.9,  0.5,   0.}

\definecolor{webred}{rgb}{0.75,0,0}
\definecolor{webgreen}{rgb}{0,0.75,0}
\usepackage[citecolor=webgreen,colorlinks=true,linkcolor=webred]{hyperref}

\newtheorem{theorem}{Theorem}[section]
\newtheorem{proposition}[theorem]{Proposition}
\newtheorem{lemma}[theorem]{Lemma}

\theoremstyle{definition}

\theoremstyle{remark}
\newtheorem{remark}[theorem]{Remark}

\newcommand{\N}{\mathbb{N}}
\newcommand{\R}{\mathbb{R}}

\newcommand{\cC}{\mathcal{C}}

\newcommand{\cE}{\mathcal{E}}
\newcommand{\cF}{\mathcal{F}}
\newcommand{\cH}{\mathcal{H}}
\newcommand{\cL}{\mathcal{L}}

\newcommand{\cO}{\mathcal{O}}

\newcommand{\cS}{\mathcal{S}}

\newcommand{\cV}{\mathcal{V}}

\newcommand{\gq}{\mathfrak{q}}

\newcommand{\gh}{\mathfrak{h}}
\newcommand{\ga}{\mathfrak{a}}

\newcommand{\gS}{\mathfrak{S}}

\newcommand{\dP}{\mathbb{P}}

\newcommand{\tr}{\widetilde{r}}

\newcommand{\rd}{\mathrm{d}}

\newcommand\spec{\gS}
\newcommand\sign{\operatorname{sign}}
\newcommand\Ran{\operatorname{Ran}}

\newcommand\dom{\operatorname{Dom}}
\newcommand\supp{\operatorname{supp}}

\newcommand{\beq}{\begin{equation}}
\newcommand{\eeq}{\end{equation}}
\newcommand{\bel}{\begin{equation}\label}
\newcommand{\ee}{\end{equation}}
\newcommand{\bea}{\begin{eqnarray}}
\newcommand{\eea}{\end{eqnarray}}
\newcommand{\beas}{\begin{eqnarray*}}
\newcommand{\eeas}{\end{eqnarray*}}

\numberwithin{equation}{section}

\begin{document}

\title[Characterization of bulk states]{Characterization of bulk states in one-edge quantum Hall systems}

\author[P.\ D.\ Hislop]{Peter D.\ Hislop}
\address{Department of Mathematics,
    University of Kentucky,
    Lexington, Kentucky  40506-0027, USA}
\email{peter.hislop@uky.edu}

\author[N.\ Popoff]{Nicolas Popoff}
\address{Aix Marseille Universit\'e, Universit\'e de Toulon, CNRS, CPT UMR 7332, 13288, Marseille, France}
\email{Nicolas.Popoff@cpt.univ-mrs.fr}

\author[E.\ Soccorsi]{Eric Soccorsi}
\address{Aix Marseille Universit\'e, Universit\'e de Toulon, CNRS, CPT UMR 7332, 13288, Marseille, France}
\email{eric.soccorsi@univ-amu.fr}


\begin{abstract}
We study magnetic quantum Hall systems in a half-plane with Dirichlet boundary conditions along the edge. Much work has been done on the analysis of the currents associated with states whose energy is located between Landau levels. These \emph{edge states} carry a non-zero current that remains well-localized in a neighborhood of the boundary.
In this article, we study the behavior of states with energies close to a Landau level. Such states are referred to as \emph{bulk states} in the physics literature. Since magnetic Schr\"odinger operator is invariant with respect to translations along the edge, it is a direct integral of operators indexed by a real wave number. We analyse these fiber operators and prove new asymptotics on the band functions and their first derivative as the wave number goes to infinity. We apply these results to prove that the current carried by a bulk state is small compared to the current carried by an edge state. We also prove that the bulk states are small near the edge.
\end{abstract}

\maketitle \thispagestyle{empty}

\tableofcontents

\vspace{.2in}

\noindent {\bf  AMS 2000 Mathematics Subject Classification:} 35J10, 81Q10,
35P20.\\
\noindent {\bf  Keywords:}
Two-dimensional Schr\"odinger operators, constant magnetic field. \\


\section{Introduction}

Quantum Hall systems consist of independent electrons constrained to open regions $\Omega$ in the plane $\R^2:=\{ (x,y),\ x, y \in \R \}$ subject to a transverse magnetic field $B(x,y) = (0,0,b(x,y))= \nabla \times a$, and possibly an electric potential $V$. The quantum Hamiltonian is $H(a,V) = (-i \nabla - a)^2 +V$ acting on a dense domain in $L^2 (\Omega)$ with self-adjoint boundary conditions. Several articles describe the physics of such systems when $\Omega$ is bounded. The analysis distinguishes between edge and bulk behavior for the states associated with the Hamiltonian, see for example \cite{Ha82,AkAvNaSe98} and \cite{HorSmi02} for a longer review. This behavior is captured by in two model domains: the plane and a half-plane, 
modeling the interior or the boundary of such a bounded system, respectively.  

In the first case, the plane model is the Landau model $\Omega = \R^2$ with constant magnetic field $b(x,y) = b$. When $V=0$, the classical electron moves in a closed circular orbit of radius the size of $b^{-1/2}$. The spectrum of $H(a,0)$ is pure point with infinitely degenerate eigenvalues at the Landau levels $E_n b$, for $n = 1,2,\ldots$, where $E_n=2n-1$. In the terminology introduced below, all of the states are \emph{bulk states}. 

In the second case of the half-plane, the restriction of the Landau model to the half-plane $x>0$ (with various boundary conditions along $x=0$) has profound consequences for the spectral and transport properties of the system.
From the classical viewpoint, the edge at $x=0$ reflects the classical orbits forming a new current along the edge.
This classical current provides the heuristic insight for quantum edge currents.
\emph{Edge states} for quantum Hall systems restrained to a half-plane $\R^*_+ \times \R
:= \{ (x,y),\ x > 0 \}$ with Dirichlet, or other boundary conditions, at $x=0$ have been analyzed by several authors \cite{DeB-P,FGW,HiSoc1}. These states $\varphi$ are constructed from wave packets with energy concentration between two consecutive Landau levels. The edge current carried by these states is $\mathcal{O}(b^{1/2})$ and it is stable under a class of electric and magnetic perturbations of the Hamiltonian. Furthermore, these states are strongly localized near $x=0$.

In contrast to edge states, \emph{bulk states} are built from wave packets with at least one Landau level in their energy interval (\cite{DeB-P}, \cite[Section 7]{HorSmi02}). This article is devoted to the mathematical study of transport and localization properties of bulk states. More specifically, we prove that one may construct bulk states for which the strength of the current is much smaller than for edge states. In addition, we prove that the bulk states are spatially localized away from the edge. Both of these results are consequences of the fact that a bulk state
has its energy concentrated in the vicinity of a Landau level.
These results are consistent with the classical picture where the orbit of particles localized away from the edge are closed and bounded.

Due to the translational invariance of the system in the $y$ direction, the magnetic Hamiltonian admits a fiber decomposition and $H$ is unitarily equivalent to the multiplication operator by a family of real analytic functions either called \emph{dispersion curves} or \emph{band functions}. The presence of an edge at $x=0$ results in non constant dispersion curves, each of them being a decreasing function in $\R$. Namely, for all $n \geq 1$, the $n$-th band function decreases from infinity to $E_n$, revealing that $E_n$ is a \emph{threshold} in the spectrum of $H$.
Moreover, the transport properties of $H$ are determined from the behavior of the \emph{velocity operator}, defined as the multiplication operator by the family formed by the first derivative of the band functions. 
 It is known that any quantum state with energy concentration between two consecutive Landau levels carries a non trivial current, indeed the velocity operator is lower bounded by some positive constant in the corresponding energy interval $I$.
This condition guarantees in addition the existence of a Mourre inequality for $H$ in $I$ (see \cite{GeNi98}). Such an estimate does not hold anymore if, unlike this, the infimum of the velocity of the band function is zero, a situation occurring when there is at least one Landau level in $\overline{I}$. In this article we study the quantum states localized in such an interval and we provide and accurate upper bound for their current when their energy concentrates near a Landau level.
Moreover, for all $n \geq 1$, the $n$-th band function approaches $E_n$ as the quasi-momentum goes to infinity, but it does not reach its limits. Hence, none of
thresholds of this model is attained, and the set of quasi-momenta associated with energy levels concentrated in the vicinity of any threshold is subsequently unbounded. This has several interesting transport and dynamical consequences, such as the the delocalization of the corresponding quantum states away from the edge $x=0$ following from the phase space analysis carried out in this article.

Notice that the usual methods of harmonic approximation, requiring that thresholds be critical points of the dispersion curves,
do not apply to this peculiar framework. The same is true for several magnetic models examined in \cite{ManPur97, DeB-P, Yaf08, FouHePer11, BruMirRai13} where the band functions tend to finite limits. 
Nevertheless, there is, to our knowledge, only a very small number of articles available in the mathematical literature, studying magnetic quantum Hamiltonians at energy levels in the vicinity of  these non-attained thresholds: we refer to \cite{BruMirRai13} for the same model (with either Dirichlet or Neumann boundary conditions) as the one investigated in the present paper and to \cite{BruPof13} for some $3$-dimensional quantum system with variable magnetic field. In these two articles, the number of eigenvalues induced by some suitable electric perturbation, which accumulate below the first threshold of the system, is estimated. The method we provide in this article to study bulk states can be easily adapted to other magnetic systems where band functions tend to finite limits, as Iwatsuka models (\cite{ManPur97,ExKov00,HiSoc13,DomHisSoc13}) and 3D translationally invariant magnetic field (\cite{Rai08},\cite{Yaf08}).

\subsection{Half-plane quantum Hall Hamiltonian}
\label{SS:HPmodel}

\paragraph{Fiber decomposition and band functions}
Put $\Omega:=\R_+^* \times \R \subset \R^2$ and let the potential $a(x,y):=(0,-bx)$ generate a constant magnetic field with strength $b>0$, orthogonal to $\Omega$.
We consider the quantum Hamiltonian in $\Omega$ with magnetic potential $a$ and Dirichlet boundary conditions at $x=0$, i.e. the self-adjoint operator acting on the dense domain $C_0^\infty(\Omega)$ as $
H(b)  := (-i \nabla - a)^2$,
and then closed in $L^2(\Omega)$. Since $\cV_b H(b) \cV_b^* = b H(1)$, where the transform
\begin{equation}\label{eq:unitary1}
(\cV_b \psi) (x,y) := b^{-1 \slash 4} \psi(\tfrac{x}{b^{1 \slash 2}}, \tfrac{y}{b^{1 \slash 2}}),
\end{equation}
is unitary in $L^2(\Omega)$, we may actually chose $b=1$ without limiting the generality of the foregoing. Thus, writing $H$ instead of $H(1)$ for notational simplicity, we focus our attention on the operator
$$ H:=-\partial_{x}^2+(-i\partial_{y}-x)^2, $$
in the remaining part of this text.

Let $\cF_y$ be the partial Fourier transform with respect to $y$, i.e.
$$
\hat{\varphi}(x,k) = (\cF_y \varphi)(x,k) :=  \frac{1}{\sqrt{2 \pi}} \int_{\R} e^{-ik y} \varphi(x,y) \rd y,\ \varphi \in L^2(\Omega).
$$
Due to the translational invariance of the operator $H$ in the $y$-direction, we have the direct integral decomposition
\bel{fiber-dec}
\cF_y H \cF_y^* = \int_\R^\oplus \gh(k) \rd k,
\ee
where the 1D operator
\bel{D:ghk}
\gh(k):=-\partial_{x}^2+V(x,k),\ V(x,k):=(x-k)^2,\ x>0,\ k \in \R,
\end{equation}
acts in $L^2(\R_{+})$ with Dirichlet boundary conditions at $x=0$. The full definition of the operator $\gh(k)$, $k \in \R$, can be found in Section \ref{sec:asymptotics1}.
For all $k \in \R$ fixed, $V(.,k)$ is unbounded as $x$ goes to infinity, so $\gh(k)$ has a compact resolvent. Let $\{ \lambda_n(k),\ n \in \N^* \}$ denote the eigenvalues, arranged in non-decreasing order, of $\gh(k)$. Since all these eigenvalues $\lambda_n(k)$ with $n \in \N^*$ are simple, then each $k \mapsto \lambda_n(k)$ is a real analytic function in $\R$.
The \emph{dispersion curves} $\lambda_{n}$, $n \in \N^*$, have been extensively studied in several articles (see e.g. \cite{DeB-P}). They are decreasing functions of $k \in \R$, obeying
\bel{E:limband}
\lim_{k\to-\infty}\lambda_{n}(k)=+\infty \ \  \mbox{and}\  \ \lim_{k\to+\infty}\lambda_{n}(k)=E_{n},
\end{equation}
for all $n \in \N^*$, where $E_n:=2n-1$ is the $n$-th Landau level.
\begin{remark}
\label{R:compacite}
For further reference, we notice from \eqref{E:limband} the following useful property:
$$\lim_{n\to +\infty}\inf_{k\in \R} \lambda_{n} (k)=+\infty.$$
\end{remark}
As a consequence the spectrum $\spec(H)$ of $H$ is $\spec(H)=\overline{\cup_{n\geq1}\lambda_{n}(\R)}=[1,+\infty)$.
The Landau levels $E_{n}$, $n \in \N^*$, are thresholds in the spectrum of $H$, and they play a major role in the analysis carried out in the remaining part of this paper.

\paragraph{Fourier decomposition}

For $n \in \N^*$ and $k \in \R$, we consider a normalized eigenfunction $u_{n}(\cdot,k)$ of $\gh(k)$ associated with $\lambda_{n}(k)$. It is well known that $u_{n}(\cdot,k)$ depends analytically on $k$. We define the $n$-th generalized Fourier coefficient of $\varphi \in L^2(\Omega)$ by
\bel{def-phin}
\varphi_{n}(k):=\langle \cF_{y}\varphi(\cdot,k),u_{n}(\cdot,k) \rangle_{L^2(\R_{+})}= \frac{1}{\sqrt{2\pi}}\int_{\R_{+}}\hat{\varphi}(x,k) \overline{u_{n}(x,k)}\rd x,
\ee
and denote by $\pi_{n}$ the orthogonal projection associated with the $n$-th harmonic:
\bel{def-pin}
\pi_{n}(\varphi)(x,y):=\frac{1}{\sqrt{2\pi}} \int_{\R}e^{i y k} \varphi_n(k) u_n(x,k) \rd k,\ (x,y) \in \Omega.
\ee
In light of \eqref{fiber-dec} we have for all $\varphi \in L^2(\Omega)$:
\bel{eq:band-state1}
\varphi  = \sum_{n\geq1}\pi_{n}(\varphi) ,
\ee
and the Parseval theorem yields
\bel{eq:l2-prop1}
\| \varphi \|_{L^2 (\Omega)}^2 = \sum_{n\geq1}  \| \varphi_n \|_{L^2(\R)}^2.
\ee

For any non-empty interval $I \subset \R$, we denote by $\dP_{I}$ the spectral projection of $H$ associated with $I$. We say that the energy of a quantum state $\varphi\in L^2(\Omega)$ is concentrated (or localized) in $I$ if $\dP_{I}\varphi=\varphi$. With reference to \eqref{fiber-dec} and \eqref{def-phin} this condition may be equivalently reformulated as
\bel{E:suppalphan}
\forall  n\in \N^{*}, \supp(\varphi_{n})\subset \lambda_{n}^{-1}(I) \, .
\ee
\subsection{Edge versus bulk}

\paragraph{Current operator and link with the velocity}
Let $y$ denote the multiplier by the coordinate $y$ in $L^2(\Omega)$. The time evolution of $y$ is the Heisenberg variable
$y(t):= e^{-itH} y e^{itH}$, for all $t \in \R$. Its time derivative is the velocity and is given by
$ \frac{\rd y(t)}{\rd t} = -i [ H, y(t)]= -i e^{-i tH} [H, y] e^{itH}$.
We define the current operator as the self adjoint operator
$$ J_{y}:=-i [H,y]=-i\partial_{y}-x, $$
acting on $\dom(H)\cap \dom(y)$. The current carried by a state $\varphi$ is $\langle J_{y} \varphi,\varphi \rangle_{L^2(\Omega)}$.

Well-known computations based on the Feynman-Hellman formula (see also \cite{DeB-P, ManPur97, ExKov00} for similar formulas involving the Iwatsuka models) yield
\bel{E:Jphi}
\forall \varphi\in L^2(\Omega), \quad \langle J_{y}\pi_{n}(\varphi),\pi_{n}(\varphi) \rangle_{L^2(\Omega)}=\int_{\R} \lambda_{n}'(k)|\varphi_{n}(k)|^2 \rd k,
\ee
linking the velocity operator, defined as the multiplication operator in $\bigoplus_{n \in \N^*} L^2(\R)$ by the family of functions $\{ \lambda_{n}',\ n \in \N^* \}$, to the current operator.
\begin{remark}
\label{R:nooverloap}
It is easy to see that \eqref{E:Jphi} extends to any quantum state $\varphi \in L^2(\Omega)$ satisfying the \emph{non-overlapping condition}
$$ \forall m \neq n,\ \supp(\varphi_{m}) \cap \supp(\varphi_{n})=\emptyset, $$
as
$$ \langle J_{y} \varphi,\varphi \rangle_{L^2(\Omega)}=\sum_{n\geq1}\int_{\R} \lambda_{n}'(k) |\varphi_{n}(k)|^2 \rd k. $$
\end{remark}

\paragraph{Edge states and bulk states}

For any bounded subinterval $I\subset \spec(H)$, the spectrum of $H$,
it is physically relevant to estimate the current carried by a state with energy concentration in $I$ and to
describe the support of such a state.

Let $\varphi \in \Ran \dP_I$ be decomposed in accordance with \eqref{def-phin}--\eqref{eq:band-state1}. Since $I$ is bounded by assumption, then the set
$\{n \in \N^*, I \cap \lambda_{n}(\R) \neq \emptyset\}$ is finite by Remark \ref{R:compacite}, so the same is true for $\{ n \in \N^*,\ \pi_n(\varphi) \neq 0 \}$. As a consequence the sum in the r.h.s. of \eqref{eq:band-state1} is finite. Notice that this fact is not a generic property of fiberered magnetic Hamiltonians (see e. g. \cite{Yaf08,BruPof13}).

For all $n \in \N^*$ put
$$ X_{n,I}:=  \Ran \dP_{I}\cap \Ran \pi_{n}, $$
so we have $\Ran \dP_{I}=\bigoplus_{n\geq1}X_{n,I}$. We shall now describe the transport and localization properties of functions in $X_{n,I}$. We shall see that, depending on whether $E_{n}$ lies inside or outside $\overline{I}$, functions in $X_{n,I}$ may exhibit radically different behaviors.

Let us first recall the results of \cite{DeB-P}, corresponding to the case $E_{n} \notin \overline{I}$.
Put $c^{-}(n,I):=\inf_{I} |\lambda'_{n} \circ \lambda_{n}^{-1} |$ and $c^{+}(n,I):=\sup_{I} | \lambda'_{n} \circ\lambda_{n}^{-1} |$. Since $\lambda_{n}$ is a decreasing function, \eqref{E:limband} yields
$$ \forall k \in \lambda_{n}^{-1}(I),\quad 0<c^{-}(n,I) \leq |\lambda_{n}'(k)| \leq c^{+}(n,I)<+\infty. $$  As a consequence, the spectrum of the current operator restricted to $X_{n,I}$ is $[-c^+(n,I),-c^{-}(n,I)]$ by \eqref{E:suppalphan}-\eqref{E:Jphi}. This entails that any state $\psi \in X_{n,I}$ carries a non-trivial current:
\bel{eq:current}
\forall \psi \in X_{n,I}, \quad  |\langle J_{y} \psi , \psi \rangle_{L^2(\Omega)}| \geq c^{-}(n,I) \|\psi \|^2.
\ee
Moreover, all quantum states $\psi \in X_{n,I}$ are mainly supported in a strip $\cS$ of width $O(1)$ along the edge.

Assume that there is no threshold in $\overline{I}$, that is $\{ E_n,\ n \in \N^{*} \} \cap \overline{I}=\emptyset$, and pick $\varphi \in \Ran \dP_I$. Since $\{ n \in \N^*,\ \pi_n(\varphi) \neq 0 \}$ is finite and $\pi_{n}(\varphi) \in X_{n,I}$ is mainly supported in $\cS$ for each $n \in \N^*$, then the same is true for $\varphi=\sum_{n \in \N^*}\pi_{n}(\varphi)$. This explains why it is referred to $\varphi$ as an \emph{edge state}. Moreover, based on Remark \ref{R:nooverloap}, \cite{DeB-P}[Proposition 2.1] entails upon eventually shortening $I$, that such a state $\varphi$ carries a non void edge current:
$$\exists c(I)>0,\ \forall \varphi \in \Ran \dP_{I},\quad |\langle J_{y}\varphi,\varphi \rangle_{L^2(\Omega)}| \geq c(I) \|\varphi \|^2. $$

Let us now examine the case where $E_{n} \in \overline{I}$ for some $n \in \N^*$. For the sake of clarity we assume in addition that there is no other threshold than $E_n$ lying in $\overline{I}$, i.e.
$\{ m \in \N^*,\ E_m \in \overline{I} \} = \{ n \}$. It is apparent that the general case where several thresholds are lying in $\overline{I}$ may be easily deduced from the single threshold situation by superposition principle.

Put $I_{n}^{-}:=(-\infty,E_{n})$ and $I_{n}^{+}=I\cap (E_{n},+\infty)$. Since $\lambda_{n}(\R)\cap I_{n}^{-}=\emptyset$, we have $\pi_{n}\circ \dP_{I_{n}^{-}}=0$, whence $X_{n,I}=X_{n,I_{n}^{+}}$. Thus it suffices to consider an energy interval $I$ of  the form
\bel{D:Idelta}
I_{n}(\delta):=(E_{n},E_{n}+\delta),\ \delta\in (0,2).
\ee
For further reference we define $k_n(\delta)$ as the unique real number satisfying
\bel{D:kdeltan}
\lambda_n(k_n(\delta)) = E_n + \delta.
\ee
Its existence and uniqueness is guaranteed by \eqref{E:limband} and the monotonicity of the continuous function $k \mapsto \lambda_n(k)$.

For all $m \neq n$, it is clear from the above analysis that $X_{m,I_{n}(\delta)}$ is made entirely of edge states.

However, this is not true for $X_{n,I_{n}(\delta)}$.
Indeed, since $\inf_{I_n(\delta)} | \lambda_n' \circ \lambda_n^{-1} |=0$, $c^{-}(n,I_n(\delta))=0$ and the bottom of the spectrum of the current operator restricted to $X_{n,I_n(\delta)}$ is zero, so that \eqref{eq:current} does not hold anymore for all $\psi \in X_{n,I_n(\delta)}$. This indicates the presence of quantum states in $X_{n,I_n(\delta)}$ carrying an arbitrarily small edge current. It turns out that such a state has part of its support localized away from the edge (a fact that will be rigorously established in this paper) and it is called a \emph{bulk state} in physical literature. We thus refer to $X_{n,I_{n}(\delta)}$ as a \emph{bulk space} and subsequently write
$X_{n,\delta}^{\rm b}$ instead of $X_{n,I_n(\delta)}$.

Notice that a definition of bulk states based on another approach is proposed in \cite{DeB-P}: De Bi\`evre and Pul\'e say that a state $\varphi$ is a bulk state associated with the non-rescaled Hamiltonian $H(b)$ when $\pi_{n}(\varphi)=\varphi$ and $\varphi_{n}(k)$ is supported in an interval of the form $(b^{\gamma},+\infty)$ whith $\gamma>1/2$. After stating our results we will come back in Section \ref{SS:avecb} to the non-rescaled Hamiltonian and we will show that our approach is more general than the one of \cite{DeB-P} and cover their results.

\begin{remark}
It is worth to mention that there are actual edge states lying in $X_{n,\delta}^{\rm b}$. With reference to \eqref{def-phin}--\eqref{eq:band-state1}, this can be seen upon noticing that any $\varphi = \pi_n(\varphi) \in X_{n,\delta}^{\rm b}$ such that $\varphi_n$ is compactly supported in $(k_n(\delta),+\infty)$, satisfies an inequality similar to \eqref{eq:current}.
\end{remark}

Let us now stress that any $\varphi \in X^{\rm b}_{n,\delta}$ expressed as
\bel{E:decomposepurebulk}
\varphi(x,y)=\int_{k_{n}(\delta)}^\infty e^{i y k} \varphi_n(k) u_n(x,k) \rd k,
\ee
where $\varphi_{n}\in L^2((k_{n}(\delta),+\infty))$ is defined by \eqref{def-phin} and satisfies
\bel{E:normealphan}
\| \varphi \|_{L^2 (\Omega)}^2 =  \int_{k_{n}(\delta)}^{+\infty} | \varphi_n(k)|^2 \rd k,
\ee
according to \eqref{eq:l2-prop1}.
Further, recalling \eqref{E:Jphi}, the current carried by $\varphi \in X_{n,\delta}$ has the following expression:
\bel{E:currentbulk}
\langle J_{y}\varphi,\varphi \rangle_{L^2(\Omega)}=\int_{k_{n}(\delta)}^{+\infty} \lambda_{n}'(k)|\varphi_{n}(k)|^2 \rd k.
\ee

\paragraph{Main goal}
In view of exhibiting pure bulk behavior, we investigate $X_{n,\delta}$ as $\delta$ goes to 0. We firstly aim to compute a suitable upper bound on the current carried by quantum states lying in $X_{n,\delta}^{\rm }$, as $\delta \downarrow 0$. Secondly we characterize the region in the half-plane where such states  are supported.

Since $k_{n}(\delta)$ tends to $\infty$ as $\delta \downarrow 0$ from \eqref{E:limband}, it is apparent that the analyis of \eqref{E:currentbulk} requires accurate
  asymptotic expansions of $\lambda_{n}(k)$ and $\lambda_n'(k)$ as $k$ goes to infinity. 

Actually, it is well known from \cite{DeB-P} or \cite[Section 2]{FouHePer11} that each $\lambda_{n}$, for $n \in \N^*$, decreases super-exponentially fast to $E_{n}$ as $k$ goes to $+\infty$:
$$ \forall \alpha\in (0,1), \exists (k_{n,\alpha}, C_{n,\alpha}) \in \R_+^* \times \R_+^*,\quad k \geq k_{n,\alpha} \Longrightarrow |\lambda_n(k)-E_{n}| \leq C_{n,\alpha} e^{-\alpha k^2/2}.
$$
 A similar upper bound on $|\lambda_{n}'|$ can be found in \cite{DeB-P}, but it turns out that these estimates are not as sharp as the one required by the analysis developped in this paper. Notice that the asymptotics of $\lambda_{n}(k)$ as $k$ tends to infnity was already investigated in \cite{Popoff}[Chapter 1] and in the unpublished work \cite{Ivrii} (the asympotics of the first derivative is derived as well in the last reference). All the above mentioned results are covered by the one obtained in this article. Notice that the asymptotics of $\lambda_n'(k)$ as $k$ tends to infinity is obtained from those calculated for the eigenfunctions associated with $\lambda_{n}(k)$ as $k \to+\infty$. Moreover these asymptotics on the eigenfunctions are useful when describing the geometrical localization of bulk states when $\delta\downarrow0$.

\subsection{Main results and outline}\label{subsec:main-results1}

Our first result is a precise asymptotics of the band functions and its derivative when $k\to+\infty$:
\begin{theorem}
\label{thm-asymptotics}
For every $n \in \N^*$ there is a constant $\gamma_n>0$ such that the two following estimates
\begin{enumerate}[i)]
\item $\lambda_{n}(k)=E_{n}+2^{2n-1} \gamma_{n}^2k^{2n-1}e^{-k^2}\left(1+\cO(k^{-2})\right)$,
\item $\lambda_{n}'(k)=-2^{2n} \gamma_{n}^2k^{2n}e^{-k^2}\left(1+\cO(k^{-2})\right)$,
\end{enumerate}
hold as $k$ goes to $+\infty$.
\end{theorem}

\begin{remark}
Notice that the second part of Theorem \ref{thm-asymptotics} may actually be recovered upon formally differentiating the first part with respect to $k$.
\end{remark}

The method used in the derivation of Theorem \ref{thm-asymptotics} is inspired by the method of quasi-modes used in \cite[Section 5]{BolHe93}. Moreover, as detailed in Subsection \ref{subsec:iwatsuka1}, the computation of the asymptotics of  $\lambda_{n}(k)$
is closely related to the rather tricky problem of understanding the eigenvalues of the Schr\"odinger operator with double wells $-h^2 \partial_t^2 + (|t|-1)^2$ in the semi-classical limit $h \downarrow 0$.

Let us now characterize bulk states with energy concentration near the Landau level $E_n$, in terms of the distance $\delta>0$ of their energy to $E_n$. We recall that $X_{n,\delta}^{\rm b}$ denotes the linear space of quantum states with energy in the interval $(E_{n},E_{n}+\delta)$ and all Fourier coefficients uniformly zero, except for the $n$-th one.

Firstly we give the smallness of the current carried by bulk states when $\delta$ goes to 0:

%

\begin{theorem}
\label{thm-bulk}
For every $n \in \N^*$ we may find two constants $\mu_n>0$ and $\delta_{n}>0$, both of them depending only on $n$, such that for each $\delta \in (0,\delta_n)$ and all $\varphi \in X^{\rm b}_{n,\delta}$, we have
\bel{i1}
\left| \langle J_{y}\varphi,\varphi\rangle \right| \leq \left( 2\delta \sqrt{|\log \delta|} + \mu_{n} \frac{\delta \log|\log\delta|}{\sqrt{|\log\delta|}} \right) \| \varphi \|_{L^2(\Omega)}^2.
\ee

\end{theorem}
\begin{remark}
\label{R:encadrement}
Estimate \eqref{i1} is accurate in the sense that for all $0 < \delta_1 < \delta_2<2$ and any interval $I_{n}:=(E_{n}+\delta_{1},E_{n}+\delta_{2})$ avoiding $E_{n}$, we find
by arguing in the exact same way as in the derivation of Theorem \ref{thm-bulk} that
$$\forall \varphi\in X_{n,I_{n}}, \quad c_{n}(\delta_{1})\| \varphi \|_{L^2(\Omega)}^2 \leq \left| \langle J_{y}\varphi,\varphi\rangle \right| \leq c_{n}(\delta_{2})\| \varphi \|_{L^2(\Omega)}^2, $$
provided $\delta_{1}$ and $\delta_{2}$ are sufficiently small. Here $c_n(\delta_j)$, $j=1,2$, stands for the constant obtained by substituting $\delta_j$ for $\delta$ in the prefactor of
$\| \varphi \|_{L^2(\Omega)}^2$ in the r.h.s. of \eqref{i1}.
\end{remark}


Finally, we discuss the localization of the bulk states in $X^{\rm b}_{n,\delta}$. We prove that when $\delta$ goes to 0, there are small in a strip of width $\sqrt{|\log \delta|}$, showing that thay are not localized near the boundary.

\begin{theorem}
\label{thm-bulk-loc1}
Fix $n \in \N^*$. Then for any $\epsilon\in (0,1)$ there exists $\delta_{n}(\epsilon)>0$, such that for all $\delta \in (0,\delta_{n}(\epsilon))$, the estimate
\bel{eq:bulk-local1}
\int_0^{(1-\epsilon)\sqrt{|\log\delta|}} \| \varphi (x,\cdot) \|_{L^2(\R)}^2 \rd x \leq C_n \epsilon^{2n-1} \delta^{\epsilon^2} |\log\delta|^{\frac{2n-1}{2}(1-\epsilon^2)} \|\varphi\|_{L^2(\Omega)}^2
\ee
holds for every $\varphi\in X_{n,\delta}^{\rm b}$ and some positive constant $C_n$ depending only on $n$.
\end{theorem}

Let $\varphi(t):=e^{-i t H \varphi}$, for $t \in \R$, be the time evolution of $\varphi \in X_{n,\delta}^{\rm b}$. Since $\varphi(t) \in X_{n,\delta}^{\rm b}$ for all $t \in \R$
it is apparent that Theorems \ref{thm-bulk} and \ref{thm-bulk-loc1} remain valid upon substituting $\varphi(t)$ for $\varphi$ in \eqref{i1}-\eqref{eq:bulk-local1}. As a consequence, the localization property and the upper bound on the current carried by a state lying in $X_{n,\delta}^{\rm b}$ survive for all times.

\subsection{Influence of the magnetic field strength}
\label{SS:avecb}
All our results are stated for the rescaled Hamiltonian $H$ with unit magnetic field strength $b=1$. We discuss the corresponding results
    for the magnetic Hamiltonian $H(b)$ associated with a constant magnetic field of strength $b>0$. Recall that $\cV_{b}$ is the unitary transformation defined in \eqref{eq:unitary1} implementing the $b^{1/2}$-scaling that allows us to normalized $H(b)$. 
Let $J_{y}(b):=-i\partial_{y}-bx$ be the non-rescaled current operator, then we have
$\cV_{b}J_{y}(b)\cV^{*}_{b}=b^{1 \slash 2} J_{y} $ and
$$\forall \varphi\in \dom(J_{y}(b)), \quad \langle J_{y}(b)\varphi,\varphi \rangle = b^{1/2}\, \langle J_{y}\cV_{b}\varphi,\cV_{b}\varphi \rangle  \, . $$
Moreover $\varphi\in X_{n,bI_{n}(\delta)}$ if and only if $\cV_{b}\varphi\in X_{n,\delta}$, therefore applying Theorem \ref{thm-bulk} we get
$$\forall \varphi\in X_{n,bI_{n}(\delta)}, \quad \left| \langle J_{y}(b)\varphi,\varphi \rangle \right| \leq b^{1 \slash 2}\,  \left( 2\delta \sqrt{|\log \delta|} + \mu_{n} \frac{\delta \log|\log\delta|}{\sqrt{|\log\delta|}} \right) \| \varphi \|_{L^2(\Omega)}^2,
 $$
as soon as $\delta$ is small enough. Looking at quantum Hall systems with strong magnetic field, it is then natural to consider $b$ large and $\delta=\delta(b)$ going to 0 as $b$ goes to $+\infty$ and one may use our result to describe the states of such a system. It is now natural to provide another possible definition for edge and bulk states associated with the non-rescaled Hamiltonian $H(b)$, in relation with the magnetic field strength: given an interval $I\subset \R$, a state $\varphi\in X_{n,I}$ is an edge state if its current is of size $b^{1 \slash 2}$, and it is a bulk state if its current is $o(b^{1 \slash 2})$ as $b$ gets large.

 We now compare our results with the analysis of De Bi\`evre and Pul\'e (\cite{DeB-P}). They define a bulk state $\varphi$ associated with the magnetic Hamiltonian $H(b)$ as a state satisfying $\pi_{n}(\varphi)=\varphi$ and such that $\varphi_{n}(k)$ is supported in an interval of the form $(b^{\gamma},+\infty)$ with $\gamma>1/2$. Put $\gamma=1/2+\epsilon$, such a state is localized in energy in the interval $(bE_{n},b\lambda_{n}(b^{\epsilon}))$. Using Theorem \ref{thm-asymptotics}, we get when $b$ gets large $\lambda_{n}(b^{\epsilon})-E_{n}\sim e^{-2b^{2\epsilon}}$. Therefore their approach is covered by the one presented in this article by setting $\delta(b):=e^{-2b^{2\epsilon}}$ and letting $b$ going to $+\infty$. For this particular choice of $\delta(b)$, Theorem \ref{thm-bulk} provides a better estimate than in \cite{DeB-P}[Corollary 2.1]. Moreover our approach is more general, in the sense that we do not restrict our analysis by a particular choice of $\delta$.

Similarly, Theorem \ref{thm-bulk-loc1} implies that any states localized in energy in the interval $(bE_{n},bE_{n}+\delta b)$ is small in a strip of width $b^{-1/2}$ when $\delta$ becomes small, with a control given by the r.h.s. of \eqref{eq:bulk-local1}.

\section{Asymptotics of the band functions}\label{sec:asymptotics1}

In this section, we prove the first part of Theorem \ref{thm-asymptotics} on the asymptotic expansion of the band functions $\lambda_n(k)$ as $k \rightarrow \infty$. The proof consists of 4 steps.
We first recall results on the harmonic oscillator and its eigenfunctions. We next construct approximate eigenfunctions
$f_n(x,k)$ of $\gh(k)$ so that $\gh(k) f_n(x,k) = E_n f_n(x,k) + R_n(x,k)$ and estimate the norm $\| R_n( \cdot,k)\|_{L^2(\R^*_+)}$.
We prove that the energy $\eta_n(k) := \langle  \gh(k) f_n ( \cdot,k), f_n(\cdot,k) \rangle$
of the approximate eigenfunction $f_n$ is a good approximation to the Landau level $E_n$. Finally, we use the Kato-Temple inequality ot obtain the result.

Here are some notations and definitions. Let us define the quadratic form
\begin{equation}
\label{D:fqgh}
\gq_{k}[u]:=\int_{\R_{+}}(|u'(x)|^2+(x-k)^2|u(x)|^2) \rd x,\ \dom(\gq_k):= \{u\in H_{0}^{1}(\R_{+}^{*}), xu \in L^2(\R_{+}^{*}) \}.
\end{equation}
Here $H_0^1(\R_+^*)$ is as usual the closure of $C_0^\infty(\R_+^*)$ in the topology of the first order Sobolev space $H^1(\R_+^*)$. The operator $\gh(k)$ (expressed in \eqref{D:ghk}) is the Friedrichs extansion of the above quadratic form and it its domain is
\begin{equation}
\label{D:domhk}
\dom(\gh(k)):= \{u \in H_0^1(\R_+^*) \cap  H^2(\R_{+}^{*}), x^2u\in L^2(\R_{+}^{*}) \} \, .
\end{equation}


\subsection{Getting started: recalling the harmonic oscillator}

The harmonic oscillator
$$
h:=-\partial_{x}^2+x^2,\ x\in \R,
$$
has a pure point spectrum made of simple eigenvalues $\{ E_n := 2n-1,\ n \in \N^* \}$, the Landau levels. The associated $L^2(\R)$-normalized eigenfunctions are the Hermite functions
\begin{equation}
\label{D:Psi}
\Psi_{n}(x):=P_{n}(x)e^{-x^2/2},\ x \in \R,\ n \in \N^*
\end{equation}
where $P_{n}$ stands for the $n$-th Hermite polynomial obeying $\deg(P_{n})=n-1$. These functions satisfy $\Psi_n (-x) = (-1)^{n-1} \Psi_n (x)$. The explicit expression \eqref{D:Psi} results in the two following asymptotic formulae (see \cite{AbSt64} or \cite{Si75})
\begin{equation}
\label{A:Psi} \Psi_{n}(x)\underset{x\to -\infty}{=}\gamma_{n}2^{n-1}x^{n-1}e^{-x^2/2}\left(1+\cO(x^{-2})\right)
\end{equation}
and
\begin{equation}
\label{A:Psi'}
\Psi_{n}'(x)\underset{x\to -\infty}{=}\gamma_{n}2^{n-1}x^{n}e^{-x^2/2}\left(-1+\cO(x^{-2})\right),
\end{equation}
where $\gamma_n:=(2^{n-1}(n-1)!\sqrt{\pi})^{-1/2}$ is a normalization constant.
Next, put
\begin{equation}
\label{D:Phi}
\Phi_{n}(x):=\Psi_{n}(x)\int_{0}^{x}|\Psi_{n}(t)|^{-2}\rd t,\ x \in \R,\ n \in \N^*,
\end{equation}
so $\{ \Psi_{n},\Phi_{n} \}$ forms a basis for the space of solutions to the ODE $hf=E_{n}f$.
Then we get
\begin{equation}
\label{A:Phi}
\Phi_{n}(x)\underset{x\to -\infty}{=}(\gamma_{n}2^{n})^{-1}\frac{e^{x^2/2}}{x^{n}}\left(1+\cO(x^{-2})\right)
\end{equation}
and
\begin{equation}
\label{A:Phi'}
 \Phi_{n}'(x)\underset{x\to -\infty}{=}(\gamma_{n}2^{n})^{-1}\frac{e^{x^2/2}}{x^{n-1}}\left(1+\cO(x^{-2})\right),
 \end{equation}
through elementary computations based on \eqref{A:Psi}--\eqref{D:Phi}.


\subsection{Building quasi-modes for $\gh(k)$ in the large $k$ regime}\label{SS:quasimode}

Following the idea of \cite{BolHe93} and \cite{Bol92} we now build quasi-modes for the operator $\gh(k)$ when the parameter $k$ is taken sufficiently large. We look at vectors of the form
\bel{qm1}
f_{n}(x,k)=\alpha(k)\Psi_{n}(x-k)+\beta(k)\chi(x,k)\Phi_{n}(x-k),\ x >0,\ k \in \R,
\ee
where $\Psi_{n}$ and $\Phi_{n}$ are respectively defined by \eqref{D:Psi} and \eqref{D:Phi}, and $\alpha$, $\beta$ are two functions of $k$ we shall make precise below. Bearing in mind
that $\Phi_{n}(\cdot,k)$ is unbounded on $\R_+^*$, the cut-off function $\chi$ is chosen in such a way that $f(\cdot,k) \in L^2(\R_{+}^*)$. Namely, we pick a non-increasing function $\chi_{0}\in \cC^{\infty}(\R_{+},[0,1])$ such that $\chi_{0}(x)=1$ for $x \in [0,\frac{1}{2}]$ and $\chi(x)=0$ for $x\in [\frac{3}{4},+\infty)$, and put
$$\chi(x,k):=\chi_{0}\left(\frac{x}{k}\right),\ x>0,\ k \in \R.$$
We impose Dirichlet boundary condition at $x=0$ on $f_{n}(\cdot,k)$, getting
$$ \beta(k)=-\alpha(k) \frac{\Psi_{n}(-k)}{\Phi_{n}(-k)}, $$
since $\Phi_{n}(-k)$ is non-zero for $k$ sufficiently large, by \eqref{A:Phi}. From this, \eqref{A:Psi} and \eqref{A:Phi}, it then follows that
\begin{equation}
\label{E:beta0}
\beta(k)= 2^{2n-1}\gamma_{n}^2 \alpha(k) k^{2n-1}e^{-k^2}\left(1+\cO(k^{-2})\right),
\end{equation}
which entails
$\|f_{n}(\cdot,k)\|_{L^2(\R_+^*)}^2=\alpha(k)^2\left(1+\cO(k^{2n-1}e^{-k^2})\right)$,
through direct computation. As a consequence we have
\bel{A:alphak}
\alpha(k)=1+\cO(k^{2n-1}e^{-k^2})
\ee
by compliance with the normalization condition $\|f_{n}(\cdot,k)\|_{L^2(\R_{+}^*)}=1$, hence
\begin{equation}
\label{E:beta1}
\beta(k)=2^{2n-1} \gamma_{n}^2 k^{2n-1} e^{-k^2}\left(1+\cO(k^{-2})\right),
\end{equation}
according to \eqref{E:beta0}.


\subsection{Energy estimation}\label{subsec:energy-est1}

Bearing in mind that $f_{n}(0,k)=0$ and $f_{n}(x,k)=\alpha(k)\Psi_{n}(x-k)$ for $x \geq 3k \slash 4$, it is clear from \eqref{D:domhk} that $f_{n}(\cdot,k) \in \dom(\gh(k))$, so
the energy carried by the state $f_n(\cdot,k)$ is well defined by
\bel{qm1b}
\eta_{n}(k):=\langle \gh(k)f_{n}(\cdot,k),f_{n}(\cdot,k) \rangle_{L^2(\R_{+}^*)}.
\ee
To estimate the error of approximation of $E_n$ by $\eta_{n}(k)$, we introduce
\bel{qm2}
r_{n}(x,k):=(\gh(k)-E_n) f_{n}(x,k),\ x>0,
\ee
in such a way that $\eta_{n}(k)-E_{n}=\langle r_{n}(\cdot,k),f_{n}(\cdot,k)\rangle_{L^2(\R_{+}^*)}$. Integrating by parts twice successively in this integral and
remembering \eqref{qm1}, we find out that
\bea
\eta_{n}(k)-E_{n} &=& \beta(k)\big\langle (\gh(k)-E_{n})(\chi(\cdot,k)\Phi_{n}(\cdot-k)),f_{n}(\cdot,k)\big\rangle_{L^2(\R_{+}^*)}
\label{E:interaction+integral} \\
& = & -\beta(k)\Phi_{n}(-k)f_{n}'(0,k)+\beta(k)\big\langle \chi(\cdot,k)\Phi_{n}(\cdot-k),r_{n}(\cdot,k)\big\rangle_{L^2(\R_{+}^*)}.
\nonumber
\eea
Further, upon combining \eqref{qm1} and \eqref{qm2} with the commutator formula $[\gh(k),\chi]=-\chi'' - 2 \chi' \partial_x$, we get that
\begin{equation}
\label{E:exprr}
r_{n}(x,k)=-\beta(k)\chi''(x,k)\Phi_{n}(x-k)-2\beta(k) \chi'(x,k)\Phi_{n}'(x-k),\ x>0,
\end{equation}
showing that $r_{n}(\cdot,k)$ is supported in $\supp (\chi'(\cdot,k))$, ie
\begin{equation}
\label{E:suppr}
\supp(r_{n}(\cdot,k))\subset [\tfrac{k}{2},\tfrac{3k}{4}].
\end{equation}
Putting \eqref{A:Phi}, \eqref{A:Phi'}, \eqref{E:beta1} and \eqref{E:exprr} together, and taking into account that
\begin{equation}
\label{E:estimchiderive}
\|\chi'(\cdot,k)\|_{L^{\infty}(\R)}=\cO(1/k) \ \  \mbox{and} \ \ \|\chi''(\cdot,k)\|_{L^{\infty}(\R)}=\cO(1/k^2),
\end{equation}
we obtain for further reference that
\begin{equation}
\label{E:estimater}
\| r_{n}(\cdot,k)\|^2_{L^2(\R_{+}^*)}=\cO(k^{2n-1}e^{-\frac{7k^2}{4}}).
\end{equation}
Let us now prove that the interaction term $-\beta(k)\Phi_{n}(0,k)f_{n}'(0,k)$ is the main contribution to the r.h.s. of  \eqref{E:interaction+integral}. Applying \eqref{A:Phi} and \eqref{A:Phi'}, we get
$$\|\Phi_{n}(\cdot-k)\|_{L^{\infty}(\frac{k}{2},\frac{3k}{4})}=\cO(k^{-n}e^{k^2/8})\ \ \mbox{and}\ \ \|\Phi_{n}'(\cdot-k)\|_{L^{\infty}(\frac{k}{2},\frac{3k}{4})}=\cO(k^{-n+1}e^{k^2/8}),$$
which, together with \eqref{E:beta0}, \eqref{E:exprr} and \eqref{E:estimchiderive}, yields
$\|r_{n}(\cdot,k) \|_{L^{\infty}(\R_+^*)}=\cO(k^{n-1}e^{-\frac{7k^2}{8}})$.
From this, \eqref{E:beta1} and the estimate
$$\left|\langle \chi(\cdot,k) \phi_{n}(\cdot-k), r_{n}(\cdot,k)\rangle_{L^2(\R_+^*)} \right|\leq \frac{k}{4}\|r_{n}(\cdot,k)\|_{L^{\infty}(\R_+^*)}\|\Phi_{n}(\cdot-k)\|_{L^{\infty}(\frac{k}{2},\frac{3k}{4})}, $$
then it follows that
$\beta(k)\langle \chi(\cdot,k) \phi_{n}(\cdot-k), r_{n}(\cdot,k)\rangle_{L^2(\R_+^*)} = \cO(k^{2n-1}e^{-\frac{7k^2}{4}})$.
Hence we have
\bel{qm3}
\eta_{n}(k)-E_{n} = -\beta(k)f_{n}'(0,k)\Phi_{n}(-k)+\cO(k^{2n-1}e^{-\frac{7k^2}{4}}),
\ee
by \eqref{E:interaction+integral}. In order to evaluate the remaining term $-\beta(k)f_{n}'(0,k)\Phi_{n}(-k)$, we take advantage of the fact that $\chi_{k}(0)=1$ and $\chi_{k}'(0)=0$, and derive from \eqref{A:Psi'} and  \eqref{A:Phi'}-\eqref{qm1} that
\bel{E:evalf'0}
f_{n}'(0,k)= \alpha(k)\Psi_{n}'(-k)+\beta(k)\Phi_{n}'(-k) = (-1)^{n-1} 2^{n} \gamma_{n} k^{n}e^{-k^2/2} \left(1+\cO(k^{-2}) \right).
\ee
Therefore we have $-\beta(k) f_{n}'(0,k)\Phi_{n}(-k)=2^{2n-1} \gamma_{n}^2 k^{2n-1} e^{-k^2} \left(1+\cO(k^{-2}) \right)$ by \eqref{A:Phi} and \eqref{E:beta1}, so we end up getting
\begin{equation}
\label{E:etaestim}
\eta_{n}(k)-E_{n} =2^{2n-1} \gamma_{n}^2k^{2n-1}e^{-k^2}\left(1+\cO(k^{-2}) \right),
\end{equation}
with the aid of \eqref{qm3}.

\subsection{Asymptotic expansion of $\lambda_n(k)$}
Let us first introduce the error term
$$\epsilon_{n}(k):=\|(\gh(k)-\eta_{n}(k))f_{n}(\cdot,k)\|_{L^2(\R_{+}^*)},$$
and combine the estimate $\epsilon_{n}(k) \leq \|r_{n}(\cdot,k)\|_{L^2(\R_{+})}+|\eta_{n}(k)-E_{n}|$ arising from \eqref{qm2}, with \eqref{E:estimater} and \eqref{E:etaestim}. We obtain that
\begin{equation}
\label{E:epsilonestim}
\epsilon_{n}(k)^2=\cO(k^{2n-1}e^{-\frac{7k^2}{4}}).
\end{equation}
We are now in position to apply Kato-Temple's inequality (see \cite[Theorem 2]{Har78}), which can be stated as follows.
\begin{lemma}
\label{L:temple}
Let $A$ be a self-adjoint operator acting on a Hilbert space $\cH$. We note $\ga$ the quadratic form associated with $A$. Let $\psi \in \dom(A)$ be $\cH$-normalized and put $\eta=\ga[\psi]$ and $\epsilon=\| (A-\eta)\psi \|_\cH$. Let $\alpha<\beta$ and $\lambda \in \R$ be such that $\spec(A)\cap(\alpha,\beta)=\{\lambda\}$. Assume that $\epsilon^2 <(\beta-\eta)(\eta-\alpha)$. Then we have
$$
\eta-\frac{\epsilon^2}{\beta-\eta}<\lambda<\eta+\frac{\epsilon^2}{\eta-\alpha}.
$$
\end{lemma}
Fix $N \in \N^*$. Since $\lim_{k \rightarrow +\infty} \lambda_{n}(k)=E_n$ for all $n \in \N^*$, we may choose $k_N>0$ so large that $\lambda_{n}(k) \in (E_{n},E_{n}+1]$ for all $k \geq k_N$ and $n\in [|1,N+1|]$. This entails
\begin{equation}
\label{E:spectralgap}
|\lambda_{n}(k)-\lambda_{p}(k)| \geq 1,\ k \geq k_{N},\ p \neq n,\ n \in  [|1,N|].
\end{equation}
Moreover, upon eventually enlarging $k_N$, we have
\bel{qm4}
|\eta_{n}(k)-(E_{n}\pm1)|\geq \tfrac{1}{2},\ k \geq k_N,\ n\in [|1,N|],
\ee
in virtue of \eqref{E:etaestim}. Thus, applying Lemma \ref{L:temple} with $\eta=\eta_{n}(k)$, $\alpha=E_{n}-1$, $\beta=E_{n}+1$ and $\epsilon=\epsilon_{n}(k)$ for each
$n \in [|1,N|]$ and $k \geq k_N$, there is necessarily one eigenvalue of $\gh(k)$ belonging to the interval $(\eta_{n}(k)-2\epsilon_{n}^2(k),\eta_{n}(k)+2\epsilon_{n}^2(k))$, according to \eqref{qm4}. Since the only eigenvalue of $\gh(k)$, $k \geq k_N$, lying in $(E_{n},E_{n}+1]$ is $\lambda_{n}(k)$, we obtain that
\begin{equation}
\label{E:avecepsilon}
 |\lambda_{n}(k)-\eta_{n}(k)| \leq 2\epsilon_{n}^2(k),\ k\geq k_{N}, n\in [|1,N|].
\end{equation}
Putting this together with \eqref{E:etaestim} and \eqref{E:epsilonestim} we end up getting the first part of Theorem \ref{thm-asymptotics}.

\subsection{Relation to a semiclassical Schr\"odinger operator and to the Iwatsuka model}\label{subsec:iwatsuka1}

In this section we exhibit the link between the asymptotics of the eigenpairs of $\gh(k)$ for large $k$ and the semi-classical limit of a Schr\"odinger operator on $\R$ with a symmetric double-wells potential.

Let us introduce the operator $H(k):=-\partial_{x}^2+(|x|-k)^2$ acting on $L^2(\R)$ and denote by $\mu_{n}(k)$ its $n$-th eigenvalue. The operator $H(k)$ is the fiber of the magnetic Laplacian associated with the Iwatsuka magnetic field $B(x,y)=\sign(x)$ defined on $\R^2$. This Hamiltonian has been studied in \cite{PeetRej, DomHisSoc13}.  The eigenfunction associated to $\mu_{n}(k)$ are even when $n$ is odd and odd when $n$ is even, therefore the restriction to $\R_{+}$ of any eigenfunction associated with $\mu_{2n}(k)$ is an eigenfunction for the operator $\gh(k)$ associated with $\lambda_{n}(k)$ and we have $\mu_{2n}(k)=\lambda_{n}(k)$. In the same way we prove that $\mu_{2n-1}(k)$ is the $n$-th eigenvalue of the operator $\gh^{\rm N}(k):=-\partial_{x}^2+(x-k)^2$ acting on $L^2(\R_{+})$ with a Neumann boundary condition.

We refer to \cite[Proposition 1.1]{Popoff} or \cite{DomHisSoc13} for more details on the link between $H(k)$, $\gh(k)$ and the operator $\gh^{\rm N}(h)$.

Using the scaling $t=kx$ we get that $H(k)$ is unitary equivalent to the operator
$$h^{-1}\left(-h^2\partial_{t}^2+(|t|-1)^2 \right), \quad t\in \R$$
where we have set $h=k^{-2}$. Therefore when $k$ gets large we reduce the problem to the understanding of the eigenvalues of the Schr\"odinger operator $-h^2\partial_{t}^2+(|t|-1)^2 $ in the semi-classical limit $h\downarrow 0$. The asymptotic expansion of the eigenvalues of Schr\"odinger operators is well-known when the potential has a unique non-degenerated minimum and uses the ``harmonic approximation", see for example \cite{Si82}. However in our case the potential $(|t|-1)^2$ is even and have a double-wells and one may expect tunneling effect between the two wells $t=1$ and $t=-1$. More precisely the eigenvalues clusters into pairs exponentially close to the eigenvalue associated to the one-wells problem that are the Landau levels (see \cite{Har80}, \cite{CoDuSe83} or \cite{HeSj84}).

The asymptotic behavior of the gap between eigenvalues in such a problem is given in \cite{HeSj84} under the hypothesis that the potential is $\cC^{\infty}(\R)$. Helffer and Sj\"ostrand use a BKW expansion of the eigenfunctions far from the wells. The key point is a pointwise estimate of an interaction term involving among others the high order derivatives of the potential at 0. Here it is not possible to use their result since the potential $(|t|-1)^2$ is not $\cC^1$ at 0. Our proof uses the fact that the potential is piecewise analytic and the knowledge of the solutions of the ODE associated to the eigenvalue problem.

 Note that mimicking the above proof it is possible to get the asymptotic expansion of the eigenvalues of the operator $\gh^{\rm N}(k)$ for large $k$ as in \cite[Section 1.4]{Popoff}.


\section{Asymptotics of the derivative of the band functions}\label{sec:asymptotics-deriv1}

In this section, we prove the asymptotic expansion of $\lambda_n^\prime (k)$, the second part of Theorem \ref{thm-asymptotics}.

\subsection{Hadamard formula}
We turn now to establishing Part ii) of Theorem \ref{thm-asymptotics}. To this purpose we introduce a sequence $\{ u_{n}(\cdot,k),\ n \in \N^* \}$ of $L^2(\R_+^*)$-normalized eigenfunctions of $\gh(k)$, verifying
$$
\left\{ \begin{array}{rcl}
-u_{n}''(x,k)+(x-k)^2u_{n}(x,k)&=& \lambda_{n}(k)u_{n}(x,k),\ x>0
\\
u_{n}(0,k)&=& 0. \end{array}
\right.
$$
Since the operator $\gh(k)$ is self-adjoint with real coefficients we choose all the $u_n(\cdot,k)$ to be real. 
Due to the simplicity of $\lambda_n(k)$, each $u_n(.,k)$ is thus uniquely defined, up to the multiplicative constant $\pm 1$.
We note $\Pi_{n}(k): \varphi\mapsto \langle \varphi(\cdot), u_n(\cdot,k) \rangle_{L^2(\R_+^*)}u_{n}(\cdot,k)$ the spectral projection of $\gh(k)$ associated with $\lambda_{n}(k)$
and call $F_n(k)$ the eigenspace spanned by $u_n(\cdot,k)$.

The proof of the asymptotic expansion of $\lambda_n'$ stated in Theorem \ref{thm-asymptotics} relies on the Hadamard formula (see \cite[Section VI]{CouHi} or \cite{Jo67}):
\begin{equation}
\label{E:formuleHadamard}
\lambda_{n}'(k)=-u_{n}'(0,k)^2,\ k \in \R,
\end{equation}
and thus requires that $u_n'(\cdot,k)$ be appropriately estimated at $x=0$. We proceed as in the derivation of \cite[Proposition 2.5]{HeSj84}.


\subsection{$H^1$-estimate of the eigenfunctions}\label{subsec:H1-est1}

The method boils down to the fact
that the operator $\gh(k)-\lambda_{n}(k)$ is a boundedly invertible on $F_{n}(k)^{\perp}$. Hence $(\gh(k)-\lambda_{n}(k))^{-1}$ is a bounded isomorphism from $F_{n}(k)^{\perp}$ onto
$\dom(\gh(k)) \cap F_n(k)^\perp$ and there exists $k_n>0$ such that we have
$$ \|(\gh(k)-\lambda_{n}(k))^{-1}\|_{\cL(F_{n}(k)^{\perp})} \leq 1,\ k \geq k_n, $$
in virtue of \eqref{E:spectralgap}. From this and the identity
$$
(\gh(k)-\lambda_{n}(k))\left(f_{n}(\cdot,k)-\Pi_{n}(k)f_{n}(\cdot,k)\right)=r_{n}(\cdot,k)+(E_{n}-\lambda_{n}(k))f_{n}(\cdot,k),
$$
arising from \eqref{qm2}, it then follows that
\begin{equation}
\label{E:estimnormef-pi}
\|f_{n}(\cdot,k)-\Pi_{n}(k)f_{n}(\cdot,k)\|_{L^2(\R_+^*)} \leq \|r_{n}(\cdot,k)\|_{L^2(\R_+^*)}+|E_{n}-\lambda_{n}(k)|,\ k \geq k_n.
\end{equation}
Moreover we have
\beas
& & \gq_{k}[f_{n}(\cdot,k)-\Pi_{n}(k) f_{n}(\cdot,k)] \\
&=& (\gq_{k}-\lambda_{n}(k))[f_{n}(\cdot,k)-\Pi_{n}(k) f_{n}(\cdot,k)]+\lambda_{n}(k)\|f_{n}(\cdot,k)-\Pi_{n}(k)f_{n}(\cdot,k)\|_{L^2(\R_+^*)}^2
\\
&=& (\gq_{k}-\lambda_{n}(k))[f_{n}(\cdot,k)]+\lambda_{n}(k) \|f_{n}(\cdot,k)-\Pi_{n}(k)f_{n}(\cdot,k)\|_{L^2(\R_+^*)}^2
\\
&=& \eta_{n}(k)-\lambda_{n}(k)+\lambda_{n}(k)\|f_{n}(\cdot,k)-\Pi_{n}(k)f_{n}(\cdot,k)\|_{L^2(\R_+^*)}^2,
\eeas
from \eqref{qm1b}, hence
\begin{equation}
\label{E:estimgq}
\gq_{k}[f_{n}(\cdot,k)-\Pi_{n}f_{n}(\cdot,k)]^{1/2} =\cO\big(\epsilon_{n}(k)+\|r_{n}(\cdot,k)\|_{L^2(\R_+^*)}+|E_{n}-\lambda_{n}(k)|\big),
\end{equation}
according to \eqref{E:avecepsilon} and \eqref{E:estimnormef-pi}. Since $\dom(\gh(k))$ (endowed with the natural norm $\gq_k[\cdot]^{1 \slash 2}$) is continuously embedded in $H^1(\R_+^*)$, we may substitute $\| f_{n}(\cdot,k)-\Pi_{n}(k)f_{n}(\cdot,k)\|_{H^1(\R_{+}^*)}$ for $\gq_{k}[f_{n}(\cdot,k)-\Pi_{n}f_{n}(\cdot,k)]^{1/2}$ in the lhs of \eqref{E:estimgq}. Thus we obtain
\bel{E:estimH1}
\| f_{n}(\cdot,k)-\Pi_{n}(k)f_{n}(\cdot,k)\|_{H^1(\R_{+}^*)}= \cO(k^{n-\frac{1}{2}}e^{-\frac{7k^2}{8}}),
\ee
with the help of \eqref{E:estimater} and Part i) in Theorem \ref{thm-asymptotics}.
As a consequence we have
\bel{qm4b}
|1-\|\Pi_{n}(k)f_{n}(\cdot,k)\|_{L^2(\R_+^*)}| \leq \|f_{n}(\cdot,k)-\Pi_{n}(k)f_{n}(\cdot,k)\|_{L^2(\R_+^*)}=\cO(k^{n-\frac{1}{2}}e^{-\frac{7k^2}{8}}),
\ee
whence
\bel{qm5}
\left\| f_{n}(\cdot,k)-\frac{\Pi_{n}(k)f_{n}(\cdot,k)}{\|\Pi_{n}(k)f_{n}(\cdot,k)\|_{L^2(\R_+^*)}}\right\|_{H^1(\R_{+}^*)}=\cO(k^{n-\frac{1}{2}}e^{-\frac{7k^2}{8}}).
\ee
Bearing in mind that
\begin{equation}
\label{E:approxu-f}
u_{n}(\cdot,k)= \frac{\Pi_{n}(k)f_{n}(\cdot,k)}{\|\Pi_{n}(k)f_{n}(\cdot,k)\|_{L^2(\R_+^*)}},
\end{equation}
upon eventually substituting $(-u_n(\cdot,k))$ for $u_n(\cdot,k)$, it follows from \eqref{qm5} that the quasi-mode $f_n(\cdot,k)$ is close to the eigenfunction $u_{n}(\cdot,k)$ in the $H^1$-norm sense, provided $k$ is large enough.
We summarize these results in the following propostion.

\begin{proposition}\label{prop:H1-est1}
For large $k$, the eigenfunction $u_n(x,k)$ is well approximated by the quasi-mode $f_n(x,k)$ in the sense that.
$$
\|f_{n}(\cdot,k)- u_n(\cdot, k) \|_{H^1(\R_{+}^*)} = \cO(k^{n-\frac{1}{2}}e^{-\frac{7k^2}{8}}) .
$$
In terms of the quadratic form $\gq_k$ defined in \eqref{D:fqgh}, it follows that
\bel{qm5a}
\gq_{k}[u_{n}(\cdot,k)-f_{n}(\cdot,k)] = \cO(k^{2n-1}e^{-\frac{7k^2}{4}}).
\ee
\end{proposition}

The proof of the second part of the proposition follows from \eqref{E:estimater},
the first part of Theorem \ref{thm-asymptotics}, \eqref{E:estimgq}, \eqref{qm4b} and \eqref{E:approxu-f}.


\subsection{$H^2$-estimate of the eigenfunctions}\label{subsec:H2-est1}

The $H^1$-estimate of Proposition \ref{prop:H1-est1} implies uniform pointwise approximation of $u_n(x,k)$ by $f_n(x,k)$.
 The Hadamard formula \eqref{E:formuleHadamard} requires
 a pointwise estimate of $u_n'(\cdot,k)$.
  Consequently, we need to estimate $u_n(\cdot,k)$ in the $H^2$-topology. Actually, $u_n(\cdot,k)$ being an eigenfunction of $\gh(k)$, it is enough to estimate the $H^1$-norm of $x^2 u_{n}(\cdot,k)$. The same problem was investigated in \cite[Section 5]{BolHe93} in the context of a bounded interval so the authors could take advantage of the fact that the multiplier by $x^2$ is a bounded operator. Although this is not the case in the framework of the present paper this slight technical issue can be overcomed through elementary commutator computations performed in the following subsection.


We start with the following straightforward inequality
\bea
\label{E:liennormH2normfg}
\|u_{n}''(\cdot,k)-f_{n}''(\cdot,k)\|_{L^2(\R_+^*)} & \leq &  \|\gh(k)[ u_{n}(\cdot,k)-f_{n}(\cdot,k)] \|_{L^2(\R_+^*)} \nonumber \\
 & & + \|(x-k)^2 [ u_{n}(\cdot,k)-  f_{n}(\cdot,k)]\|_{L^2(\R_+^*)}.
\eea
Recall from \eqref{qm2} that $r_n(x,k) = ( \gh(k) - E_n) f_n(x,k)$. Then, since
$$
\gh(k) [ u_{n}(\cdot,k)-f_{n}(\cdot,k) ] =\lambda_{n}(k)u_{n}(\cdot,k)-E_{n}f_{n}(\cdot,k)+r_{n}(\cdot,k),
$$
the first term on the right hand side of \eqref{E:liennormH2normfg} is bounded above by
$$
\| \gh(k)[ u_{n}(\cdot,k)-f_{n}(\cdot,k)] \|_{L^2(\R_+^*)}\leq |\lambda_{n}(k)-E_{n}|+E_{n}\|u_{n}(\cdot,k)-f_{n}(\cdot,k)\|_{L^2(\R_+^*)} + \|r_{n}(\cdot,k) \|_{L^2(\R_+^*)} .
$$
The first part of Theorem \ref{thm-asymptotics}, \eqref{E:estimater} and \eqref{E:approxu-f} then yield
\begin{equation}
\label{E:estimterme1}
 \| \gh(k)[ u_{n}(\cdot,k)-f_{n}(\cdot,k) ] \|_{L^2(\R_+^*)}=\cO(k^{n-\frac{1}{2}}e^{-\frac{7k^2}{8}}).
\end{equation}

To treat the second term on the right in \eqref{E:liennormH2normfg}, we introduce
$$ v_{n}(x,k):=(x-k)u_{n}(x,k) \quad \mbox{and}\quad g_{n}(x,k):=(x-k)f_{n}(x,k),$$
and notice that $g_{n}(\cdot,k)$ belongs to $\dom(\gh(k))$. Similarly, taking into account that $u_{n}(\cdot,k)$ decays super-exponentially fast for $x$ sufficiently large, since $\lim_{x \rightarrow +\infty} V(x,k)=+\infty$, we see that $v_{n}(\cdot,k)$ belongs to $\dom(\gh(k))$ as well. Therefore, we have
$$\gh(k)\left(v_{n}(x,k)-g_{n}(x,k)\right)=\lambda_{n}(k)v_{n}(x,k)-E_{n}g_{n}(x,k)-2(u_{n}'(x,k)-f_{n}'(x,k))+(x-k)r_{n}(x,k), $$
by straightforward computations, hence
\begin{equation}
\label{E:trouveqv-g}
\gq_{k}[v_{n}(\cdot,k)-g_{n}(\cdot,k)]\leq \|\tr_{n}(\cdot,k) \|_{L^2(\R_+^*)} \|v_{n}(\cdot,k)-g_{n}(\cdot,k)\|_{L^2(\R_+^*)},
\end{equation}
where we have set
\bel{qm5b}
\tr_{n}(x,k):=\lambda_{n}(k)v_{n}(x,k)-E_{n}g_{n}(x,k)-2(u_{n}'(x,k)-f_{n}'(x,k))+(x-k)r_{n}(x,k),\ x>0.
\ee
Evidently,
\begin{equation}
\label{E:estimordre1}
\|v_{n}(\cdot,k)-g_{n}(\cdot,k)\|_{L^2(\R_+^*)}^2 \leq \gq_{k}[u_{n}(\cdot,k)-f_{n}(\cdot,k)],
\end{equation}
so, by \eqref{qm5a},
we are left with the task of estimating the $L^2$-norm of $\tr_{n}(\cdot,k)$.  In light of \eqref{qm5b}-\eqref{E:estimordre1} and the basic estimate
$\|u_{n}'(\cdot,k)-f_{n}'(\cdot,k)\|_{L^2(\R_+^*)}^2 \leq \gq_{k}[u_{n}(\cdot,k)-f_{n}(\cdot,k)]$, we find that
\bea
\|\tr_{n}(\cdot,k)\|_{L^2(\R_+^*)}
& \leq & |\lambda_{n}(k)-E_{n}|\|g_{n}(\cdot,k)\|_{L^2(\R_+^*)}+(2+\lambda_{n}(k)) \gq_{k}[u_{n}(\cdot,k)-f_{n}(\cdot,k)]^{1/2}  \nonumber \\
& & +\| (x-k)r_{n}(\cdot,k)\|_{L^2(\R_+^*)}. \label{qm6}
\eea
Next, we pick $k_n>0$ so large that $\eta_{n}(k) =\gq_{k}[f_{n}(\cdot,k)] \leq E_n+1$ for all $k \geq k_n$, according to \eqref{E:etaestim}, so we have
\bel{qm7}
\|g_{n}(\cdot,k)\|_{L^2(\R_+^*)} \leq \gq_{k}[f_{n}(\cdot,k)] \leq E_n+1,\ k \geq k_n.
\ee
Bearing in mind that $\lambda_{n}(k)\leq \lambda_{n}(0) \leq 4n-1$ for $k \geq 0$, we deduce from \eqref{qm6}-\eqref{qm7} that
\bel{qm8}
\|\tr_{n}(\cdot,k)\|_{L^2(\R_+^*)}\leq c_n \left(|\lambda_{n}(k)-E_{n}|+\gq_{k}[u_{n}(\cdot,k)-f_{n}(\cdot,k)]^{1/2}\right)+\| (x-k)r_{n}(\cdot,k)\|_{L^2(\R_+^*)},
\ee
for all $k \geq k_n$, where $c_n$ is some positive constant depending only on $n$.
Last, recalling \eqref{E:suppr} and \eqref{E:estimater}, we get that
$$\|(x-k)r_{n}(\cdot,k)\|_{L^2(\R_+^*)}=\cO(k^{n+\frac{1}{2}}e^{-\frac{7k^2}{8}}). $$
From this, Part i) in Theorem \ref{thm-asymptotics}, \eqref{qm5a} and \eqref{qm8} then it follows that
\bel{qm9}
\|\tr_{n}(\cdot,k)\|_{L^2(\R_+^*)}=\cO(k^{n+\frac{1}{2}}e^{-\frac{7k^2}{8}}).
\ee

Now, putting \eqref{qm5a}, \eqref{E:trouveqv-g}, \eqref{E:estimordre1} and \eqref{qm9} together, we obtain
\bel{qm10}
\gq_{k}[v_{n}(\cdot,k)-g_{n}(\cdot,k)] =\cO(k^{2n}e^{-\frac{7k^2}{4}}).
\ee
Further, since $\|(x-k)^2u_{n}(\cdot,k)-(x-k)^2f_{n}(\cdot,k)\|_{L^2(\R_+^*)}^2 \leq \gq_{k}[v_{n}(\cdot,k)-g_{n}(\cdot,k)]$, we deduce from
\eqref{E:liennormH2normfg}, \eqref{E:estimterme1} and \eqref{qm10} that $\|f_{n}''(\cdot,k)-u_{n}''(\cdot,k) \|_{L^2(\R_+^*)} =\cO(k^{n}e^{-\frac{7k^2}{8}})$. We obtain the following proposition.

\begin{proposition}\label{prop:H2est1}
For all $n\in \N^{*}$, there exists $k_{n}\in \R$ and $C_{n}>0$ such that we have
\bel{eq:h2est2}
\forall k >k_{n}, \quad \|f_{n}(\cdot,k)-u_{n}(\cdot,k) \|_{H^2(\R_+^*)} \leq C_{n}k^{n}e^{-\frac{7k^2}{8}}.
\ee
\end{proposition}
Since $H^2(\R_+^*)$ is continuously embedded in $W^{1,\infty}(\R_+^*)$, we deduce for $k\geq k_{n}$:
\bel{eq:Linf-deriv-est1}
\|f_{n}'(\cdot,k)-u_{n}'(\cdot,k)\|_{L^{\infty}(\R_{+}^*)}\leq C_{n} k^{n}e^{-\frac{7k^2}{8}}
\ee
and
\bel{eq:Linf-fnc-est1}
\|f_{n}(\cdot,k)-u_{n}(\cdot,k)\|_{L^{\infty}(\R_{+}^*)}\leq C_{n} k^{n}e^{-\frac{7k^2}{8}} \, .
\ee
These results guarantee that any pointwise estimate of $u_{n}'(\cdot,k)$ on $\R_+$ is uniformly well approximated by the one of the quasi-mode $f_{n}'(\cdot,k)$, provided $k$ is large enough. More precisely, we have
\bel{qm11}
u_{n}'(0,k)=f_{n}'(0,k)+\cO(k^{n}e^{-\frac{7k^2}{8}}).
\ee
Finally, plugging \eqref{E:evalf'0} into \eqref{qm11} and then applying \eqref{E:formuleHadamard}, we obtain the second part of
Theorem \ref{thm-asymptotics}.

\begin{remark}
Higher order expansions of $\lambda_{n}(k)$ and $\lambda_{n}'(k)$ may be derived from sharper asymptotics of the Hermite functions than \eqref{A:Psi}-\eqref{A:Psi'} (see \cite[Section 1.4]{Popoff}).
\end{remark}

\section{Characterization of bulk states}\label{sec:bulk-states1}

This section is devoted to characterizing functions in the bulk space $X^{\rm b}_{n,\delta}$ as $\delta\downarrow0$. This is achieved by means of the asymptotic analysis carried out in the previous sections.

 Remember from Subsection \ref{SS:HPmodel} that $\varphi\in X^{\rm b}_{n,\delta}$ decomposes as in \eqref{E:decomposepurebulk}, that $k_{n}(\delta)$ is defined by \eqref{D:kdeltan}, that $\varphi_{n}\in L^2(k_{n}((\delta),+\infty))$ and that the current carried by $\varphi$ is given by \eqref{E:currentbulk}. The asymptotic behavior of the quasi-momentum $k_{n}(\delta)$ when $\delta\downarrow0$ is derived in Section \ref{subsec:ev-est1}.  Section \ref{subsec:proof-of-main1} and \ref{subsec:proof-of-main2} are devoted to the proof of Theorems \ref{thm-bulk} and \ref{thm-bulk-loc1}.
\subsection{Estimates on quasi-momenta associated with bulk components}\label{subsec:ev-est1}

We already know that $k_{n}(\delta)$ goes to $+\infty$ as $\delta \downarrow 0$. More precisely:
\begin{lemma}
We have the following asymptotics as $\delta \downarrow 0$:
\bel{bs1}
k_{n}(\delta)=\sqrt{|\log\delta|}+\frac{2n-1}{4}\left(\frac{\log|\log\delta|}{\sqrt{|\log\delta|}}\right)+o\left(\frac{\log|\log\delta|}{\sqrt{|\log\delta|}}\right).
\ee
\end{lemma}

\begin{proof}
Since $\lim_{\delta \downarrow 0} k_{n}(\delta)=+\infty$ by \eqref{E:limband}, we deduce from the first part of
Theorem \ref{thm-asymptotics} that
$$\gamma_{n}^22^{2n-1}k_{n}(\delta)^{2n-1}e^{-k_{n}(\delta)^2}\left(1+\cO(k_{n}(\delta)^{-2})\right)=\delta.$$
Set $\widetilde{\gamma}_{n}:=\log(\gamma_{n}2^{2n-1})$. Taking the logarithm of both sides of this identity we find
\begin{equation}
\label{E:logkdelta}
\widetilde{\gamma}_{n}+(2n-1)\log(k_{n}(\delta))-k_{n}(\delta)^2=\log \delta +\cO(k_{n}(\delta)^{-2}),
\end{equation}
showing that
\bel{bs0}
k_{n}(\delta)\underset{\delta\downarrow 0}{\sim}\sqrt{|\log \delta|}.
\ee
Plugging this into \eqref{E:logkdelta}, we get
\beas
 k_{n}^2(\delta)&=& -\log \delta+(2n-1)\log\left(\sqrt{|\log \delta|}+o(\sqrt{|\log\delta|}) \right)+\widetilde{\gamma}_{n}+\cO(k_{n}(\delta)^{-2})
 \\
 &=& -\log \delta+\frac{2n-1}{2}\log\left(-\log \delta)\right)+\widetilde{\gamma}_{n}+o(1),
\eeas
which entails \eqref{bs1}.
\end{proof}
Notice that the first-order term in the above asymptotic expansion of $k_{n}(\delta)$ as $\delta \downarrow 0$, is independent of $n$.


\subsection{Asymptotic velocity and proof of Theorem \ref{thm-bulk}}\label{subsec:proof-of-main1}

In light of \eqref{bs1} we may estimate the asymptotics of $\lambda_{n}'(k(\delta))$ as $\delta \downarrow 0$. We combine both parts of Theorem \ref{thm-asymptotics}, getting,
$$ \frac{\lambda_{n}'(k)}{\lambda_{n}(k)-E_n}=-2k  \left(1+\cO(k^{-2})\right)$$
and then substitute $k_n(\delta)$ (resp., the r.h.s. of \eqref{bs1}) for $k$ in the lhs (resp., the r.h.s.) of this identity. Bearing in mind that $\lambda_n(k_n(\delta)) = E_n + \delta$, we obtain
$$
\lambda_{n}'(k_{n}(\delta))=-2\delta \sqrt{|\log \delta|}-\frac{2n-1}{2}\left(\frac{\delta \log|\log\delta|}{\sqrt{|\log\delta|}}\right)+ o\left(\frac{\delta \log|\log\delta|}{\sqrt{|\log\delta|}}\right).
$$
Similarly to \eqref{bs1} it turns out that the first order term in this expansion does not depend on the energy level $n$.

Let us now upper bound $(-\lambda_{n}'(k))$ in the interval $(k_{n}(\delta),+\infty)$ with the following:
\begin{lemma}
\label{lm-derivative}
Let $n\in \N^*$. Then there are two constants $\delta_n>0$ and $\mu_{n}>0$, such that the estimate
$$ 0  \leq -\lambda_{n}'(k) \leq 2\delta \sqrt{|\log \delta|} + \mu_{n} \frac{\delta \log|\log\delta|}{\sqrt{|\log\delta|}}, $$
holds for all $\delta \in (0,\delta_n)$ and all $k \geq k_n(\delta)$.
\end{lemma}
\begin{proof}
 From the second part of Theorem \ref{thm-asymptotics}, we may find two constants $\tilde{k}_n>0$ and $c_n>0$, depending only on $n$, such that we have
\bel{bs2}
\forall k \geq \tilde{k}_n,\ 0 \leq -\lambda_{n}'(k) \leq  2^{2n} \gamma_{n}^2 k^{2n}e^{-k^2} \left( 1+\frac{c_n}{k^{2}} \right).
\ee
With reference to \eqref{bs0}, we choose $\delta_n>0$ so small that $k_n(\delta_n) \geq \tilde{k}_n$. We get
\bel{bs3}
\forall \delta \in (0,\delta_n),\ \forall k \geq k_n(\delta),\ 0 \leq -\lambda_{n}'(k) \leq 2^{2n} \gamma_{n}^2 k^{2n}e^{-k^2} \left( 1 + \frac{c_n}{k^{2}} \right),
\ee
from \eqref{bs2}. Further, $k \mapsto 2^{2n} \gamma_{n}^2 k^{2n}e^{-k^2} (1+ c_n \slash k^{2})$ being a decreasing
function on $[\sqrt{n},+\infty)$, it follows from \eqref{bs3}, upon eventually shortening $\delta_n$ so that $k_n(\delta_n) \geq \sqrt{n}$, that
\bel{bs4}
\forall \delta \in (0,\delta_n),\ \forall k \geq k_n(\delta),\ 0 \leq -\lambda_{n}'(k) \leq 2^{2n} \gamma_{n}^2 k_n(\delta)^{2n}e^{-k_n(\delta)^2} \left( 1 + \frac{c_n}{k_n(\delta)^{2}} \right).
\ee
Due to the first part of Theorem \ref{thm-asymptotics} the r.h.s. of \eqref{bs4} is upper bounded by $2 k_n(\delta) (\lambda_n(k_n(\delta))-E_n)(1 + \tilde{c}_n \slash k_n(\delta)^2)$ for some constant $\tilde{c}_n>0$ depending only on $n$. The desired result follows from this, \eqref{bs1} and the identity $\lambda_n(k_n(\delta))-E_n=\delta$.
\end{proof}
Now Theorem \ref{thm-bulk} follows readily from \eqref{E:normealphan}, \eqref{E:currentbulk} and Lemma \ref{lm-derivative}.

\subsection{Proof of Theorem \ref{thm-bulk-loc1}}\label{subsec:proof-of-main2}

For $\epsilon \in (0,1)$ fixed, put $a_n(\delta) : = (1-\epsilon) k_n(\delta)$, where $k_{n}(\delta)$ is defined in \eqref{D:kdeltan}. Let $\varphi\in X_{n,\delta}^{\rm b}$ be a $L^2$-normalized state and define $\cE_n(\delta):=\int_{0}^{a_n(\delta)} \| \varphi(x,\cdot) \|_{L^2(\R)}^2 \rd x$. Then we have
\bel{eq:local-est0}
\int_{x=0}^{(1-\epsilon)\sqrt{|\log\delta|}}\int_{\R}|\varphi(x,y)|^2 \rd x \rd y \leq \cE_{n}(\delta),
\ee
from \eqref{bs1}, provided $\delta$ is small enough. In view of majorizing $\cE_n(\delta)$, we recall from \eqref{E:decomposepurebulk} that
$$ \| \varphi(x,\cdot) \|_{L^2(\R)} = \| \hat{\varphi}(x,\cdot) \|_{L^2(\R)} = \| \varphi_n u_n(x,\cdot) \|_{L^2(k_n(\delta),+\infty)}, $$
so we get that
\bel{eq:local-est1}
\cE_n(\delta) = \int_{k_n(\delta)}^\infty |\varphi_n(k)|^2 \| u_n(\cdot,k) \|_{L^2(0,a_n(\delta))}^2 \rd k.
\ee
Let $f_{n}(\cdot,k)$ be the quasi-mode of $\gh(k)$ introduced in Section \ref{SS:quasimode}. As
$$\| u_n(\cdot,k) \|_{L^2(0,a_n(\delta))}^2 \leq 2 \left( \| u_n(\cdot,k) - f_n(\cdot,k) \|_{L^2(\R_+^*)}^2 + \| f_n(\cdot,k) \|_{L^2(0,a_n(\delta))}^2 \right), $$
we deduce from \eqref{E:normealphan}, \eqref{eq:local-est1} and Proposition \ref{prop:H1-est1} that for every $\delta>0$ small enough, we have
\bel{eq:local-est3}
\cE_n(\delta) \leq C_n k_n(\delta)^{2n-1} e^{-7 k_n(\delta)^2 \slash 4}\| \varphi \|_{L^2(\Omega)}^2 +  2 \cF_n(\delta),
\ee
with $\cF_n(\delta):=\int_{k_n(\delta)}^{+\infty} | \varphi_n(k) |^2 \| f_n(\cdot,k) \|_{L^2(0,a_n(\delta))}^2 \rd k$. Here and henceforth, $C_n$ is some positive constant, depending only on $n$.
In virtue of \eqref{eq:local-est0} and \eqref{eq:local-est3}, we are thus left with the task of estimating $\cF_n(\delta)$ from above. To do that we use the explicit form \eqref{qm1} of the quasi-mode $f_n$, getting
\bel{eq:local-est3a}
\cF_n(\delta) \leq 2 ( \psi_n(\delta) + \phi_n(\delta)),
\ee
with
\bea
\psi_n(\delta) & := & \int_{k_n(\delta)}^\infty |\alpha(k)|^2|\varphi_{n}(k)|^2 \| \Psi_n(\cdot-k) \|_{L^2(0,a_n(\delta))}^2 \rd k \label{eq:local-est3b} \\
\phi_n(\delta) & := & \int_{k_n(\delta)}^\infty |\beta(k)|^2|\varphi_{n}(k)|^2 \| \Phi_n(\cdot-k) \|_{L^2(0,a_n(\delta))}^2 \rd k. \label{eq:local-est3c}
\eea
Bearing in mind that $k_{n}(\delta)$ tends to $+\infty$ as $\delta \downarrow 0$, we treat each of the two terms in the r.h.s. of \eqref{eq:local-est3a} separately.

Performing the change of variable $\tilde{x}=x-k$ in the r.h.s. of \eqref{eq:local-est3b} and bearing in mind that $k_{n}(\delta)$ tends to $+\infty$ as $\delta \downarrow 0$, we deduce from \eqref{A:alphak} that
\bea
\psi_{n}(\delta) & =  & \int_{k=k_{n}(\delta)}^{+\infty}\int_{\tilde{x}=-k}^{-k+(1-\epsilon)k_{n}(\delta)} |\alpha(k)|^2|\varphi_{n}(k)|^2|\Psi_{n}(\tilde{x})|^2\rd \tilde{x} \rd k
\nonumber \\
 & \leq & C_{n} \int_{k=k_{n}(\delta)}^{+\infty}\int_{\tilde{x}=-k}^{-k+(1-\epsilon)k_{n}(\delta)}|\varphi_{n}(k)|^2|\Psi_{n}(\tilde{x})|^2\rd \tilde{x} \rd k
\nonumber \\
& \leq & C_{n}\int_{\tilde{x}=-\infty}^{-\epsilon k_{n}(\delta)}\int_{k=\max(k_{n}(\delta),-\tilde{x})}^{-\tilde{x}+(1-\epsilon)k_{n}(\delta)}|\varphi_{n}(k)|^2|\Psi_{n}(\tilde{x})|^2 \rd k \rd \tilde{x}, \label{eq:local-est3d}
\eea
for $\delta$ sufficiently small. Next, recalling the normalization condition \eqref{E:normealphan}, giving $\int_{k\in \R}|\varphi_{n}(k)|^2\rd k =\|\varphi\|_{L^2(\Omega)}^2=1$, and taking $\delta>0$ so small that $\epsilon k_{n}(\delta)$ is sufficiently large in order to apply \eqref{A:Psi}, we derive from \eqref{eq:local-est3d} that
$$ \psi_{n}(\delta) \leq  C_{n} \int_{\tilde{x}=-\infty}^{-\epsilon k_{n}(\delta)} \tilde{x}^{2n}e^{-\tilde{x}^2} \rd \tilde{x}. $$
Further, taking into account that $\int_{-\infty}^{L} \tilde{x}^{m} e^{-\tilde{x}^2} \rd \tilde{x} \sim -\frac{L^{m-1}}{2}e^{-L^2}$ as $L\to-\infty$ for any $m \in \N$, we may thus find $\delta_{n}(\epsilon)>0$ so that we have
\bel{E:controlepsindelta}
\forall \delta \in (0,\delta_{n}(\epsilon)),\quad \psi_{n}(\delta)\leq C_{n} \epsilon^{2n-1} k_{n}(\delta)^{2n-1}e^{-\epsilon^2 k_{n}(\delta)^2}.
\ee
Similarly, upon substituting $\phi_n$, \eqref{E:beta1} and \eqref{eq:local-est3c} for $\psi_n$, \eqref{A:alphak} and \eqref{eq:local-est3b}, respectively, in the above reasoning, we find out for $\delta$ sufficiently small that
\bel{E:controlepsindeltab}
\phi_{n}(\delta)
\leq C_{n} \int_{\tilde{x}=-\infty}^{-\epsilon k_{n}(\delta)}\int_{k=\max(k_{n}(\delta),-\tilde{x})}^{-\tilde{x}+(1-\epsilon)k_{n}(\delta)}k^{4n-2} e^{-2k^2}|\varphi_{n}(k)|^2|\Phi_{n}(\tilde{x})|^2 \rd k \rd \tilde{x}.
\ee
Thus, taking $\delta>0$ so small that $k \mapsto k^{4n-2}e^{-2k^2}$ is decreasing for $k \geq k_n(\delta)$, we deduce from \eqref{E:controlepsindeltab} with the help of \eqref{E:normealphan}, that
$$
\phi_{n}(\delta)
\leq C_{n}  \left(
  \int_{\tilde{x}=-\infty}^{-k_{n}(\delta)} \tilde{x}^{4n-2} e^{-2 \tilde{x}^2}|\Phi_{n}(\tilde{x})|^2 \rd \tilde{x}
+ \int_{\tilde{x}=-k_{n}(\delta)}^{-\epsilon k_{n}(\delta)}k_{n}(\delta)^{4n-2} e^{-2k_{n}(\delta)^2}|\Phi_{n}(\tilde{x})|^2 \rd \tilde{x}  \right)
$$
Applying \eqref{A:Phi} we see that there exists $\delta_n(\epsilon)>0$ so small that we have
$$
\forall \delta \in (0,\delta_n(\epsilon)),\quad \phi_n(\delta) \leq C_n k_{n}(\delta)^{2n-3} e^{-k_{n}(\delta)^2}.
$$
Putting this together with \eqref{eq:local-est3}-\eqref{eq:local-est3a} and \eqref{E:controlepsindelta} we end up getting $\delta_{n}(\epsilon)>0$ such that
$$ \forall \delta \in (0,\delta_{n}(\epsilon)), \quad \cE_{n}(\delta) \leq C_{n} \epsilon^{2n-1} k_{n}(\delta)^{2n-1} e^{-\epsilon^2k_{n}(\delta)^2}. $$
Now, Theorem \ref{thm-bulk-loc1} follows from this, \eqref{bs1} and \eqref{eq:local-est0}.

\paragraph{Aknowledgement}
N. Popoff is financially supported by the ARCHIMEDE Labex (ANR-11-LABX- 0033) and the A*MIDEX project (ANR-11-IDEX-0001-02) funded by the ``Investissements d'Avenir" French government program managed by the ANR. He also thanks the University of Kentucky for its invitation in February 2014 to Lexington where this work has been finalized. P. D. Hislop thanks the  Centre de Physique Th\'eorique, CNRS, Luminy, Marseille, France, for its hospitality. P. D. Hislop was partially supported by the Universit\'e du Sud Toulon-Var, La Garde, France, and National Science Foundation grant 11-03104
during the time part of the work was done.
\bibliographystyle{abbrv}
\bibliography{biblio}

\end{document}